\documentclass[conference]{IEEEtran}

\makeatletter

\newcommand{\Rmnum}[1]{\expandafter\@slowromancap\romannumeral #1@}
\makeatother

\usepackage{graphicx,subfigure,epsfig}
\usepackage{cite,verbatim,multirow}
\usepackage{amssymb,amsmath,amsfonts}
\usepackage{amsthm}
\usepackage{mathrsfs}
\usepackage{float}
\usepackage{soul}
\usepackage{colortbl, color}
\usepackage{cases}
\usepackage[ruled,vlined,linesnumbered]{algorithm2e}

\usepackage{algpseudocode}

\usepackage{subfigure}
\usepackage{url, xcolor}
\usepackage{verbatim, booktabs}

\newtheorem{lemma}{Lemma}

\usepackage{ulem}
\newcommand{\ls}[1]  
   {\dimen0=\fontdimen6\the=#1\dimen0
    \advance\lineskip.5\fontdimen5\the\lineskip-\dimen0
    \lineskiplimit=.9\lineskip
    \baselineskip=\lineskip
    \advance\baselineskip\dimen0
    \normallineskip\lineskip
    \normallineskiplimit\lineskiplimit
    \normalbaselineskip\baselineskip
    \ignorespaces
   }

\begin{document}
\bibliographystyle{ieeetr}

\title{Interference Management in Massive MIMO HetNets: A Nested Array Approach} 

\author{\authorblockN{Mingjie Feng and Shiwen Mao}
\authorblockA{Dept. Electrical \& Computer Engineering, Auburn University, Auburn, AL 36849-5201}
Email: mzf0022@auburn.edu, smao@ieee.org}

\maketitle

\begin{abstract}
The nested array, which is implemented by nonuniform antenna placement, is an effective approach to achieve $O(N^2)$ degrees of freedom (DoF) with an antenna array of $N$ antennas. Such DoF refers to the number of directions of incoming signals that can be resolved. With the increased number of DoF, an important application of nested array is to nullify the interference signals from multiple directions. In this paper, we apply nested array in a massive MIMO heterogeneous network (HetNet) for interference management. With nested array based interference nulling, each base station (BS) can nullify a certain number of interference signals. A key design issue is to select the interference sources to be nullified at each BS. We formulate this problem as an integer programming problem. The objective is to maximize the sum rate of all users, subject to BS DoF constraints. We propose an approximation scheme to solve this problem and derive a performance upper bound. Simulation results show that the proposed scheme effectively improves the sum rate and achieves a near optimal performance.
\end{abstract}


\pagestyle{plain}\thispagestyle{plain}


\section{Introduction}

{\em Massive MIMO} (Multiple Input Multiple Output) and {\em small cell} are recognized as two key technologies for 5G wireless systems due to their great potential to enhance network capacity~\cite{Andrews14}. In a massive MIMO (also known as large-scale MIMO or very large MIMO) system, the base station (BS) is equipped with more than 100 antennas and serves multiple users with the same spectrum band~\cite{Marzetta10}. With aggressive spatial multiplexing, a massive MIMO can dramatically improve both energy and spectral efficiency compared to traditional wireless systems~\cite{Ngo13,Xu14Access,Feng16Network}. Small cell deployment achieves high signal to noise ratio (SNR) and high spectrum spatial reuse due to short transmission distance and small coverage area. As a result, a heterogeneous network (HetNet) with small cells can significantly boost network capacity compared to traditional macrocell network.

Due to these benefits, massive MIMO HetNet, which integrates these two techniques, has drawn considerable attention recently~\cite{Hosseini13,Bjornson13,Bethanabhotla14,Xu15,Feng16}. In a massive MIMO HetNet, multiple small cell BS's (SBS) coexist with a macrocell BS (MBS) equipped with a large number of antennas. Due to spectrum scarcity in cellular networks, small cells are expected to share the same spectrum band with the macrocell, resulting in cross-tier interference. While interference management in a regular HetNet mainly focuses on resource allocation in the time-frequency domain~\cite{Feng14}, the spatial characteristics of massive MIMO can be exploited to mitigate interference in massive MIMO HetNets. In~\cite{Adhikary15}, a spatial blanking scheme was proposed in which the transmission energy of MBS is focused on certain directions that do not cause interference to small cells. In~\cite{Hosseini13}, a reversed time division duplex (RTDD) architecture was introduced. 
Since the channels between MBS and SBS's are quasi-static, the MBS can carry out zero-forcing beamforming based on the estimated channel covariance. In~\cite{Bjornson13}, coordinated transmission is assumed between MBS and SBS's, so that each user receives signals from both MBS and SBS's. Through coordinated beamforming vectors for the BS's, the interference between different transmissions can be minimized.

In this paper, we consider interference management in massive MIMO HetNets from the perspective of antenna array processing techniques. The idea is to identify the directions of interference sources, and then the interference from certain number of directions can be nullified at the antenna array using a second order processing technique called {\em nested array}~\cite{Pal10}. Based on the concept of {\em difference co-array}, a nested array is implemented by nonuniform antenna placement. Nested array achieves a degree of freedom (DoF) of $O(N^2)$ with only $N$ antennas. Further, the direction of arrival estimation is performed with a {\em passive sensing} pattern, i.e., the antenna array does not need to send out signals for detection. Due to the benefits of passive sensing, easy implementation, and large number of DoF, we apply nested array to interference management in massive MIMO HetNets. With the nested array configuration at each SBS, the directions of $O(N^2)$ users can be estimated. Then, the interference from different directions can be filtered, so that the desired signal remains while interference signals are nullified.

Given the DoF of each SBS, a key design problem is to select the set of interfering users that use the DoF for interference nulling. We formulate such a problem as an integer programming problem with the objective of maximizing the sum rate of a massive MIMO HetNet. Since the objective function is nonlinear and nonconvex, we propose a series of approximations that transform the original problem into an integer programming problem with a linear objective function. The optimal solution of such a problem can be obtained with a {\em cutting plane} approach. Moreover, we find that if an additional condition is satisfied by the system, the constraint matrix is {\em unimodular} and the integer problem is equivalent to an linear programming (LP) problem obtained by relaxing the integer constraints. Thus an optimal solution can be obtained with an LP solver. To evaluate the performance of the solution obtained by the approximation, we derive a performance upper bound for comparison purpose. The proposed scheme is evaluated with simulations and compared with other benchmark schemes. The results show that near optimal performance can be achieved.

The remainder of this paper is organized as follows. The nested array based interference nulling method is introduced in Section~\ref{sec:nest}. The system model and problem formulation are presented in Section~\ref{sec:prob} and the solution is presented in Section~\ref{sec:sol}. The simulation results are discussed in Section~\ref{sec:sim}. We conclude this paper in Section~\ref{sec:con}.

\section{Preliminaries \label{sec:nest}}

\subsection{Signal Model of Difference Co-Array}

Consider an antenna array with $N$ antennas, the $N \times 1$ steering vector corresponding to direction $\theta$ is denoted as ${\bf{a}}(\theta )$. Let $d_i$ be the position of the $i$th antenna and $\lambda$ the carrier wavelength. The $i$th element of ${\bf{a}}(\theta)$ is ${e^{j\left( {2\pi /\lambda } \right){d_i}\sin \theta }}$. Suppose $D $ narrowband sources from directions $\left\{ {{\theta _i},i = 1,2,...,D} \right\}$ impinge the antenna array with powers $\left\{ {\sigma _i^2,i = 1,2,...,D} \right\}$. The received signal is given by
\begin{align}\label{eq1}
{\bf{r}} [ m ] = {\bf{F} \boldsymbol{\gamma}}[ k ] + {\bf{n}}\left[ m \right], \; m=1,2,...,N,
\end{align}
where ${\boldsymbol{\gamma}}{\left[ m \right]_{D  \times 1}} = {\left[ {{\gamma_1}\left[ m \right],{\gamma_2}\left[ m\right],..., {\gamma_D}\left[ m \right]} \right]^T}$ is the source signal vector, ${\bf{F}} = \left[ {{\bf{f}}\left( {{\theta _1}} \right),{\bf{f}}\left( {{\theta _1}} \right),..., {\bf{f}}\left( {{\theta _D}} \right)} \right]$ is the array manifold matrix, and ${\bf{n}}\left[ m \right]$ is the white noise vector. We assume that the sources are temporally uncorrelated. Hence the autocorrelation matrix of ${\boldsymbol{\gamma}}\left[ m \right]$ is diagonal. Then, the autocorrelation matrix of the received signal is given by~\cite{Hoctor90}
\begin{align}\label{eq2}
{{\boldsymbol{\Omega}}_{{\bf{rr}}}} =&\; \mathbb{E} \left[ {{\bf{r}}{{\bf{r}}^H}} \right] = {\bf{F}}{{\boldsymbol{\Omega}}_{{\boldsymbol \gamma} {\boldsymbol \gamma}}} {{\bf{F}}^H} + \sigma _n^2{\bf{I}} \nonumber\\
                       =&\; {\bf F} \left( {\begin{array}{*{20}{c}}
                       {\sigma _1^2}&{}&{}&{}\\
                       {}&{\sigma _2^2}&{}&{}\\
                       {}&{}& \ddots &{}\\
                       {}&{}&{}&{\sigma _D^2}
                                     \end{array}} \right){\bf F}^H + \sigma _n^2{\bf{I}}.
\end{align}

We next vectorize ${{\boldsymbol{\Omega}}_{{\bf{rr}}}}$ and obtain the following vector~\cite{Hoctor90}.
\begin{align}\label{eq3}
{\bf{z}} =&\; {\rm{vec}}\left( {{{\boldsymbol{\Omega}}_{{\bf{rr}}}}} \right) = {\rm{vec}}\left[ {\sum\limits_{i = 1}^D {\sigma _i^2\left( {{\bf{f}}\left( {{\theta _i}} \right){{\bf{f}}^H}\left( {{\theta _i}} \right)} \right)} } \right] + \sigma _n^2\mathop {{{\bf{1}}_n}}\limits^ \to  \nonumber\\
         =&\; \left( {{{\bf{F}}^*} \odot {\bf{F}}} \right){\bf{p}} + \sigma _n^2\mathop {{{\bf{1}}_n}}\limits^ \to,
\end{align}
where ${\bf{p}} = {\left[ {\sigma _1^2, \sigma _2^2, ..., \sigma _D^2} \right]^T}$ is the power vector of the $D$ sources. $\mathop {{{\bf{1}}_n}}\limits^ \to   = {\left[ {{\bf{e}}_1^T, {\bf{e}}_2^T, ..., {\bf{e}}_N^T} \right]^T}$, and ${{\bf{e}}_i}$ is a column vector with $1$ at the $i$th position and all other elements $0$.
Comparing \eqref{eq3} with \eqref{eq1}, we find that ${\bf{z}}$ can be regarded as a signal received at an array with manifold matrix given as ${{{\bf{F}}^*} \odot {\bf{F}}}$, where $\odot$ denotes the Khatri-Rao (KR) product. The corresponding source signal is ${\bf{p}}$ and the noise vector is given as $\sigma _n^2\mathop {{{\bf{1}}_n}}\limits^ \to  $. Analyzing the manifold matrix ${{{\bf{F}}^*} \odot {\bf{F}}}$, we find that the distinct rows of ${{{\bf{F}}^*} \odot {\bf{F}}}$ behave like the manifold of an array with antenna positions given by distinct values in the set $\left\{ {\mathop {{{\boldsymbol{\varepsilon}}_i}}\limits^ \to   - \mathop {{{\boldsymbol{\varepsilon}}_j}}\limits^ \to  ,1 \le i,j \le N} \right\}$, where ${\mathop {{{\boldsymbol{\varepsilon}}_i}}\limits^ \to  }$ is the position vector of the original array. The new array is the difference co-array of the original array~\cite{Hoctor90}.

In a difference co-array with antenna positions given in the set $\left\{ {\mathop {{{\boldsymbol{\varepsilon}}_i}}\limits^ \to   - \mathop {{{\boldsymbol{\varepsilon}}_j}}\limits^ \to  } \right\},\forall i,j = 1,2,...,N$, it is easy to see that the number of elements in this set is $N(N - 1)+1$. Thus, given the original $N$-antenna array, the maximum DoF of a difference co-array is
\begin{align}\label{eq4}
{\rm{DO}}{{\rm{F}}_{\max }} = N(N - 1).
\end{align}
We thus conclude that a DoF of $O({N^2})$ can be achieved with $N$ antennas by exploiting the second order statistics of the received signal~\cite{Hoctor90}. This opens tremendous opportunities to detect more sources than the number of physical antenna elements.

\subsection{Nested Array: An Effective Approach to Increase DoF}

Based on the difference co-array framework, the nested array was proposed in~\cite{Pal10} as an effective solution to the problem of resolving more sources than antenna elements. Nested array is characterized by {\em non-uniform antenna array placement} and second order statistic processing of the received signal. According to the analysis in~\eqref{eq3}, the difference co-array of a nested array has $O({N^2})$ antenna elements, and thus a nested array achieves a DoF of $O({N^2})$. Compared to existing methods on increasing DoF, the nested array approach is easy to implement with reduced overhead and can be applied to more general scenarios since less assumptions are needed for the system model. In addition, the nested array operates in a passive sensing pattern, which only needs to receive source signals. These favorable features make nested array suitable to applications in cellular networks. The implementation and setup process of a nested array are described in~\cite{Pal10}.

\subsection{Interference Nulling with Nested Array}

An important application of nested array is its capability of interference nulling. Let ${\bf{z}} = \left( {{{\bf{F}}^*} \odot {\bf{F}}} \right){\bf{p}} + \sigma _n^2\mathop {{{\bf{1}}_n}}\limits^ \to $ be the equivalent received signal at the difference co-array of the nested array of an SBS. Suppose we apply a beamforming with weight vector ${\bf{w}}$. Then, the resulting signal is given by
\begin{align}\label{eq5}
r' = {{\bf{w}}^H}{\bf{z}} = \sum\limits_{i = 1}^D {{{\bf{w}}^H}\left( {{{\bf{a}}^*}\left( {{\theta _i}} \right) \otimes {\bf{a}}\left( {{\theta _i}} \right)} \right)\sigma _i^2}  + \sigma _n^2{{\bf{w}}^H}\mathop {{{\bf{1}}_n}}\limits^ \to  .
\end{align}
In~(\ref{eq5}), $r'$ can be regarded as a weighted sum of ${\sigma _i^2},i=1,2,...,D$, and $\sigma _n^2$ with wights given as ${{{\bf{w}}^H}\left( {{{\bf{f}}^*}\left( {{\theta _i}} \right) \otimes {\bf{f}}\left( {{\theta _i}} \right)} \right)}$ and ${{\bf{w}}^H}\mathop {{{\bf{1}}_n}}\limits^ \to $, respectively. Define the new beam pattern as
\begin{align}\label{eq6}
B({\theta _i}) = {{\bf{w}}^H}\left( {{{\bf{f}}^*}\left( {{\theta _i}} \right) \otimes {\bf{f}}\left( {{\theta _i}} \right)} \right).
\end{align}
Thus, the powers of sources from different directions get spatially filtered by $B({\theta _i}), i=1,2,...,D$. It is then possible to adjust these new beam patterns so that the antenna array only receives desired signals, while nulling noise and interference signals.

For an SBS with nested array, suppose the directions of its SUEs are $\left\{ {{\delta  _l},l = 1,2,...,L} \right\}$, and the directions of interfering SUEs and MUEs are $\left\{ {{\eta  _i},i = 1,2,...,I} \right\}$. Then, the beam patterns of different directions are expected to be
\begin{align}\label{eq7}
\begin{array}{l}
B\left( {{\delta _l}} \right) = 1,~l = 1,2,...,L,\\
B\left( {{\eta _i}} \right) = 0,~i = 1,2,...,I.
\end{array}
\end{align}
According to the expression of ${\bf{z}}$ in~\eqref{eq3}, the beamforming weight vector should satisfy
\begin{align}\label{eq8}
\left( {\begin{array}{*{20}{c}}
{{{\left( {{{\bf{f}}^*}\left( {{\delta _1}} \right) \otimes {\bf{f}}\left( {{\delta _1}} \right)} \right)}^H}}\\
 \vdots \\
{{{\left( {{{\bf{f}}^*}\left( {{\delta _L}} \right) \otimes {\bf{f}}\left( {{\delta _L}} \right)} \right)}^H}}\\
{{{\left( {{{\bf{f}}^*}\left( {{\eta _1}} \right) \otimes {\bf{f}}\left( {{\eta _1}} \right)} \right)}^H}}\\
 \vdots \\
\begin{array}{l}
{\left( {{{\bf{f}}^*}\left( {{\eta _I}} \right) \otimes {\bf{f}}\left( {{\eta _I}} \right)} \right)^H}\\
{~~~~~~~~( {\mathop {{{\bf{1}}_n}}\limits^ \to  } )^T}
\end{array}
\end{array}} \right){\bf{w}} = \left( {\begin{array}{*{20}{c}}
1\\
 \vdots \\
1\\
0\\
 \vdots \\
\begin{array}{l}
0\\
0
\end{array}
\end{array}} \right).
\end{align}

With the solution of ${\bf{w}}$, the weight vector of the original antenna array can be determined by the method presented in~\cite{Pal10}. It can be observed from~\eqref{eq8} that the DoF to identify and manage the desired signals, noise, and interfering signals is $O({N^2})$. This enforces a constraint on the number of desired and interference sources that can be managed, and we will consider this in the problem formulation.
The nested array based interference nulling approach provides a new perspective to interference management, by employing spatial filtering on different directions. With such desirable features of nested array, it is highly promising to apply this technique to interference management in massive MIMO HetNets.

\section{Problem Formulation \label{sec:prob}}

We consider a two-tier massive MIMO HetNet consists of one MBS with massive MIMO (labeled as BS $j=0$) and multiple SBS's with regular MIMO (denoted as $j=1,2,...,J$). There are $K$ users (indexed by $k=1,2,...,K$) to be served, and we assume that the set of users served by each BS is pre-determined (e.g., through a user association algorithm~\cite{Xu15}). Let ${x_{k,j}}$ be the user association variable defined as
\begin{align}\label{eq9}
x_{k,j} \doteq \left\{ \begin{array}{ll}
	         1, & \mbox{user $k$ is associated with BS $j$} \\
					 0, & \mbox{otherwise,}
					          \end{array} \right. \nonumber \\
									k=1,2,...,K, \; j=0,1,...,J,
\end{align}
which are
predetermined for all $k$, $j$.

The macrocell and small cells share the same spectrum band and both tiers adopt the time division duplex (TDD) mode in a synchronized way, i.e., the two tiers use the same time period for uplink or downlink transmissions. The SBS's use nested array to perform interference nulling, so that the uplink interference from a certain number of users can be nulled with the beamforming process presented in~\eqref{eq8}. Since the nested array requires second order processing on all antennas, the MBS with massive MIMO adopts the traditional linear array for DoA estimation and interference management due to complexity concerns. With the DoAs of the interference links, both MBS and SBS's can perform directional transmissions to avoid downlink interference to a certain number of users. This way, the mutual interference between the BS's and some users can be eliminated. 

Define binary variables $n_{k,j}$ for interference nulling as
\begin{align}\label{eq10}
  n_{k,j} \doteq \left\{ \begin{array}{ll}
	         1, & \mbox{BS $j$ nulls interference from user $k$} \\
					 0, & \mbox{otherwise,}
					          \end{array} \right. \nonumber \\
									k=1,2,\ldots,K, \; j=0,1,\ldots,J.
\end{align}
According to our analysis on~\eqref{eq8}, $n_{k,j}$ should satisfy
\begin{align}\label{eq11}
\sum\limits_{k = 1}^K {{x_{k,j}}{q_{k,j}}}  + \sum\limits_{k = 1}^K {{n_{k,j}}{q_{k,j}}} + 1 \le {D_j},~j=0,1,...,J,
\end{align}
where $q_{k,j}$ is the number of multipath from user $k$ to BS $j$, $D_j$ is the upper bound for the number of directions that can be resolved by BS $j$. It is assumed that noise is always nulled at the BS's with one DoF.

We use the data rate model of massive MIMO HetNet in~\cite{Bethanabhotla14}. For SUEs, the sum of uplink and downlink data rate of user $k$ connecting to SBS $j$ can be approximated as
\begin{align}\label{eq12}
& R_{k,j} = \log \left( {1 + \frac{{{M_j} - {S_j} + 1}}{{{S_j}}}  \frac{{{p_k}{g_{k,j}}}}{{\sum\limits_{k' \ne k} {{p_{k'}}{g_{k',j}}\left( {1 - {n_{k',j}}} \right)} }}} \right) + \nonumber \\
&~~~~ \log \left( {1 + \frac{{{M_j} - {S_j} + 1}}{{{S_j}}}  \frac{{{p_j}{g_{k,j}}}}{{1 + \sum\limits_{j' \ne j} {{p_{j'}}{g_{k,j'}}\left( {1 - {n_{k,j'}}} \right)} }}} \right), \nonumber\\
&~~~~~k=1,2,...,K, \; j=1,...,J,
\end{align}
where $p_k$ and $p_j$ are the power of user $k$ and BS $j$, respectively; $g_{k,j}$ is the large-scale channel gain between user $k$ and BS $j$; $M_j$ is the number of antennas of BS $j$; $S_j$ is the upper bound for the number of users that can be served by BS $j$.

For a macrocell user, due to the law of large numbers in a massive MIMO system, the interference caused by other macrocell users can be averaged out. The sum of uplink and downlink data rates of user $k$ connecting to MBS is given by
\begin{align}\label{eq13}
& R_{k,0} = \log \left( {1 \hspace{-0.025in}+\hspace{-0.025in} \frac{{{M_0} \hspace{-0.025in}-\hspace{-0.025in} {S_0} \hspace{-0.025in}+\hspace{-0.025in} 1}}{{{S_0}}}  \frac{{{p_k}{g_{k,0}}}}{{\sum\limits_{j = 1}^J {\sum\limits_{k \in {\mathcal{U}_j}} {{p_{k'}}{g_{k',0}}\left( {1 \hspace{-0.025in}-\hspace{-0.025in} {n_{k',0}}} \right)} } }}} \right) \hspace{-0.025in}+ \nonumber \\
&~~~~ \log \left( {1 + \frac{{{M_0} - {S_0} + 1}}{{{S_0}}}  \frac{{{p_0}{g_{k,0}}}}{{1 + \sum\limits_{j = 1}^J {{p_j}{g_{k,j}}\left( {1 - {n_{k,j}}} \right)} }}} \right), \nonumber\\
&~~~~~k=1,2,...,K,
\end{align}
where $\mathcal{U}_j=\left\{ {k\left| {{x_{k,j}} = 1} \right.} \right\}$ is the set of UEs served by BS $j$.

The sum rate maximization of a massive MIMO HetNet is formulated as
\begin{align}
& {\bf P1:\/}  \max_{\left\{ {n_{k,j}} \right\}}  \sum_{k = 1}^K x_{k,0} R_{k,0} + \sum_{k = 1}^K \sum_{j = 1}^J x_{k,j} R_{k,j} \label{eq14}  \\
&\mbox{subject to:} \nonumber \\
&\hspace{0.1in} \sum_{k = 1}^K {{x_{k,j}}{q_{k,j}}}  + \sum_{k = 1}^K {{n_{k,j}}{q_{k,j}}} + 1 \le {D_j}, \; j=0,1,...,J \label{eq15} \\
&\hspace{0.1in} n_{k,j} \le 1 - x_{k,j}, \; k=1,2,\ldots,K, \; j=0,1,...,J \label{eq16} \\
&\hspace{0.1in} \; n_{k,j} \in \left\{ 0, 1 \right\}, \; k=1,2,\ldots,K, \; j=0,1,...,J. \label{eq17}
\end{align}
Constraint~\eqref{eq16} is due to the fact that when BS $j$ serves user $k$, it does not need to null interference from user $k$.

\section{Solution Algorithm \label{sec:sol}}

Problem {\bf P1\/} is an integer programming program with a nonlinear and non-convex objective function, which is generally NP-hard. To make the problem tractable, we assume the system operate in the high SINR regime, so that $\log \left( {1 + {\rm{SINR}}} \right) \approx \log \left( {{\rm{SINR}}} \right)$. The high SINR assumption is reasonable in a massive MIMO HetNet due to the large antenna array gain of massive MIMO and the short transmission distance of small cells. Applying this approximation to (\ref{eq12}) and (\ref{eq13}), the objective function of problem {\bf P1\/} is given as
\begin{align}\label{eq18}
\sum_{k = 1}^K \sum_{j = 1}^J x_{k,j} V_{k,j},~k=1,2,...,K, \; j=0,1,...,J,
\end{align}
where $V_{k,j}$ is given in~(\ref{eq19}) (see the next page).
\begin{figure*}
\begin{align}\label{eq19}
\begin{array}{l}
{V_{k,0}} = \log \left( {\frac{{{M_0} - {S_0} + 1}}{{{S_0}}}{p_k}{g_{k,0}}} \right) \hspace{-0.025in}+\hspace{-0.025in} \log \left( {\frac{{{M_j} - {S_j} + 1}}{{{S_j}}}{p_0}{g_{k,0}}} \right) \hspace{-0.025in}-\hspace{-0.025in} \log \left( {\sum\limits_{j = 1}^J {\sum\limits_{k \in {U_j}} {{p_{k'}}{g_{k',0}}\left( {1 \hspace{-0.025in}-\hspace{-0.025in} {n_{k',0}}} \right)} } } \right)
\hspace{-0.025in}-\hspace{-0.025in} \log \left( {1 \hspace{-0.025in}+ \hspace{-0.025in}\sum\limits_{j = 1}^J {{p_j}{g_{k,j}}\left( {1 \hspace{-0.025in}-\hspace{-0.025in} {n_{k,j}}} \right)} } \right) \\
{V_{k,j}} = \log \left( {\frac{{{M_j} - {S_j} + 1}}{{{S_j}}}{p_k}{g_{k,j}}} \right) \hspace{-0.025in}+\hspace{-0.025in} \log \left( {\frac{{{M_j} - {S_j} + 1}}{{{S_j}}}{p_j}{g_{k,j}}} \right) \hspace{-0.025in}-\hspace{-0.025in} \log \left( {\sum\limits_{k' \ne k} {{p_{k'}}{g_{k',j}}\left( {1 - {n_{k',j}}} \right)} } \right) \hspace{-0.025in}-\hspace{-0.025in} \log \left( {1 \hspace{-0.025in}+\hspace{-0.025in} \sum\limits_{j' \ne j} {{p_{j'}}{g_{k,j'}}\left( {1 \hspace{-0.025in}-\hspace{-0.025in} {n_{k,j'}}} \right)} } \right). 
\end{array}
\end{align}
\vspace{-0.3in}
\end{figure*}
We remove the constants in (\ref{eq19}) and apply the property $\sum_i {\log {x_i}}=\log \left( {\prod_i {{x_i}} } \right)$. Since $\log ( \cdot )$ is a monotonic function, {\bf P1\/} can be transformed into the following problem.
\begin{align}
   {\bf P2:\/} &\; \max_{\left\{ {n_{k,j}} \right\}}  \prod\limits_{j = 0}^J {\prod\limits_{k \in {\mathcal{U}_j}} {{W_{k,j}}} }  \label{eq20}  \\
   \mbox{subject to:} &\;\; (\ref{eq15}), (\ref{eq16}), \mbox{and } (\ref{eq17}), \nonumber
\end{align}
where
\begin{align}\label{eq21}
&{W_{k,0}} = \nonumber\\
&\left[ {\sum\limits_{j = 1}^J {\sum\limits_{k \in {\mathcal{U}_j}} {{p_{k'}}{g_{k',0}}\left( {1 - {n_{k',0}}} \right)} } } \right]\left[ {1 + \sum\limits_{j = 1}^J {{p_j}{g_{k,j}}\left( {1 - {n_{k,j}}} \right)} } \right], \nonumber\\
&\;\; k=1,2,...,K.
\end{align}
\begin{align}\label{eq22}
&{W_{k,j}} = \nonumber\\
&\left[ {\sum\limits_{k' \ne k} {{p_{k'}}{g_{k',j}}\left( {1 - {n_{k',j}}} \right)} } \right]\left[ {1 + \sum\limits_{j' \ne j} {{p_{j'}}{g_{k,j'}}\left( {1 - {n_{k,j'}}} \right)} } \right] \nonumber\\
&\;\; k=1,2,...,K, \; j=1,...,J.
\end{align}

It can be seen that the objective function of {\bf P2\/} is a product of linear expressions, which can be expressed as a polynomial on the set of variables $\left\{ {{n_{k,j}}} \right\}$. Thus, {\bf P2\/} is nonlinear integer programming with a complicated form, which is hard to solve with normal approaches. However, we can make use of a property of 0-1 problems to approximate problem {\bf P2\/} with a linear integer programming problem. Then, the cutting plane method~\cite{Gomory58} can be employed to effectively obtain the optimal solution of the linear integer programming problem.

\subsubsection{Linear Approximation of {\bf P2\/}}

Consider the product of 0-1 variables, with all the variables having the same probability distribution. When the number of variables is increased, the product becomes less likely to be $1$ since it is less likely that all the variables are $1$. As the objective function of {\bf P2\/} is a weighted sum of products of 0-1 variables, the values of higher-order parts are more likely to be $0$. Thus, the impact of the higher-order parts is limited. Let $P$ be the probability that an arbitrary ${n_{k,j}}$ equals to $1$. In the objective function of {\bf P2\/}, the probability for an $M$-th order product to be $1$ is $P^M$.
\begin{lemma}\label{lemma1}
$P$ can be approximated by $\frac{\bar{D}_j - 1}{{\bar q}_{k,j} K} - \frac{1}{J} - \frac{1}{{\bar q}_{k,j} J K}$, where $\bar{z}$ is the mean of a variable $z$.
\end{lemma}
\begin{proof}
To maximize the sum rate, all DoFs of each BS are expected to be used for data transmission and interference nulling. Thus, all the constraints described by (\ref{eq15}) are close to equality. Adding these equations from $j=0$ to $j=J$, we have $\sum_{k = 1}^K {\sum_{j = 0}^J {{n_{k,j}}{q_{k,j}}} }  = \sum_{j = 0}^J {{D_j}}  - \sum_{k = 1}^K {\sum_{j = 0}^J {{x_{k,j}}{q_{k,j}}} }  - J - 1$. The probability that ${n_{k,j}}$ equals to $1$ can be derived as
\begin{align}
P =&\; \frac{{\sum\limits_{k = 1}^K {\sum\limits_{j = 0}^J {{n_{k,j}}{q_{k,j}}} } }}{{\sum\limits_{k = 1}^K {\sum\limits_{j = 0}^J {{q_{k,j}}} } }} = \frac{{\sum\limits_{j = 0}^J {{D_j}} \hspace{-0.025in}-\hspace{-0.025in} \sum\limits_{k = 1}^K {\sum\limits_{j = 0}^J {{x_{k,j}}{q_{k,j}}} }  \hspace{-0.025in}-\hspace{-0.025in} J \hspace{-0.025in}-\hspace{-0.025in} 1}}{{\bar q}_{k,j} J K} \nonumber\\
=&\; \frac{\sum\limits_{j = 0}^J {D_j} \hspace{-0.025in}-\hspace{-0.025in} {\bar q}_{k,j} K \hspace{-0.025in}-\hspace{-0.025in} J \hspace{-0.025in}-\hspace{-0.025in} 1}{{\bar q}_{k,j} J K}
= \frac{{\bar D}_j \hspace{-0.025in}-\hspace{-0.025in} 1}{{\bar q}_{k,j} K} \hspace{-0.025in}-\hspace{-0.025in} \frac{1}{J} \hspace{-0.025in}-\hspace{-0.025in} \frac{1}{{\bar q}_{k,j} J K}. \nonumber
\end{align}
\end{proof}

In a typical cellular network, the number of users in a macrocell can be more than 500, i.e., $K>500$. The DoFs are $D_j=O (N^2)$ for SBS's and $D_0=O ( N )$ for the MBS. As the typical number of antennas for MBS and SBS are 100 and 10, respectively, we have ${\bar D}_j \approx 100$. Therefore, the value of $P$ is expected to be small in a practical system, and the values of higher-order terms of $P$ are close to $0$. Due to this fact, we make linear approximations for the higher-order parts in the polynomial of (\ref{eq20}) and transform the objective of {\bf P2\/} to a linear function.

Let ${\bf{\tilde n}}$ be the vector concatenating the columns of matrix ${\left[ {{n_{k,j}}} \right]_{K \times J}}$. For a product with $M$ elements in ${\bf{\tilde n}}$ denoted as ${{{\bf{\tilde n}}}_{{i_1}}},{{{\bf{\tilde n}}}_{{i_2}}}, \cdots {{{\bf{\tilde n}}}_{{i_M}}}$, we have the following approximation.
\begin{align}\label{eq24}
{{{\bf{\tilde n}}}_{{i_1}}}{{{\bf{\tilde n}}}_{{i_2}}} \cdots {{{\bf{\tilde n}}}_{{i_M}}} = \frac{{{P^{M - 1}}}}{M}\left( {{{{\bf{\tilde n}}}_{{i_1}}} + {{{\bf{\tilde n}}}_{{i_2}}} +  \cdots  + {{{\bf{\tilde n}}}_{{i_M}}}} \right).
\end{align}
It can be easily verified that the expectations of both sides are equal to $P^M$, thus the long term performance of the problem with approximation equals to the performance of original problem. Considering that the product value is close to $0$ when $M$ is large enough, this approximation is expected to be accurate.

With the linear approximation of the polynomial objective function, {\bf P2\/} is transformed to the following integer programming problem.
\begin{align}
  {\bf P3:\/} \max_{\bf{\tilde n}} & \;\; {\bf{c}} {\bf{\tilde n}} \label{eq25}\\
  \mbox{subject to:}   & \;\; {\bf{A}} {\bf{\tilde n}} \le {\bf{b}}.\label{eq26}
\end{align}
The vector $\bf{c}$ is determined by applying the linear transformation of \eqref{eq24} to \eqref{eq20}. The constraint matrix $\bf{A}$ is given by

\begin{align}\label{eq27}
& {{\bf{A}}_{\left( {K + 1} \right)\left( {J + 1} \right) \times  K \left( {J + 1} \right)}} \doteq \left( {\begin{array}{*{20}{c}}
{\bf{Q}}\\
{\bf{I}}
\end{array}} \right).
\end{align}
$\bf{I}$ is a $K\left( {J + 1} \right) \times K\left( {J + 1} \right)$ identity matrix. $\bf{Q}$ is given by
\begin{align}\label{eq28}
{{\bf{Q}}_{\left( {J + 1} \right) \times K\left( {J + 1} \right)}} = \left( {\begin{array}{*{20}{c}}
{{{\bf{q}}_0}}&{\bf{0}}& \cdots &{\bf{0}}\\
{\bf{0}}&{{{\bf{q}}_1}}& \cdots &{\bf{0}}\\
 \vdots & \vdots & \ddots & \vdots \\
{\bf{0}}&{\bf{0}}& \cdots &{{{\bf{q}}_J}}
\end{array}} \right),
\end{align}
where
${{\bf{q}}_j} = \left[ {{q_{1,j}},{q_{2,j}}, ...,{q_{K,j}}} \right],~j = 0,1,...,J$.
The vector $\bf{b}$ is given by
\begin{align}\label{eq30}
&{\bf{b}}_{\left( {K + 1} \right) \left( {J + 1} \right) \times 1} \doteq  \\
&\left[ {{E_0},...,{E_J},1 \hspace{-0.025in}-\hspace{-0.025in} {x_{1,0}}, ..., 1 \hspace{-0.025in}-\hspace{-0.025in} {x_{K,0}}, 1 \hspace{-0.025in}-\hspace{-0.025in} {x_{1,1}}, ..., 1 \hspace{-0.025in}-\hspace{-0.025in} {x_{K,J}}} \right]^T, \nonumber
\end{align}
where ${E_j} = {D_j} - \sum_{k = 1}^K {{x_{k,j}}{q_{k,j}}} - 1, ~j = 0, 1,..., J$.

\subsubsection{Performance Upper Bound}

To verify the effectiveness of the approximation, we derive a performance upper bound for {\bf P2\/} and compare it with the proposed scheme in Section~\ref{sec:sim}. Consider a linear approximation given by
\begin{align}\label{eq32}
{{\bf{\tilde n}}_{{i_1}}}{{\bf{\tilde n}}_{{i_2}}} \cdots {{\bf{\tilde n}}_{{i_M}}} = \frac{{{{{\bf{\tilde n}}}_{{i_1}}} + {{{\bf{\tilde n}}}_{{i_2}}} +  \cdots  + {{{\bf{\tilde n}}}_{{i_M}}}}}{M}.
\end{align}
The resulting integer programming can be expressed as
\begin{align}
  {\bf P4:\/} \max_{\bf{\tilde n}} & \;\; {\bf{c}}' {\bf{\tilde n}} \label{eq33}\\
  \mbox{subject to:}   & \;\; {\bf{A}} {\bf{\tilde n}} \le {\bf{b}},\label{eq34}
\end{align}
where ${\bf{c}}'$ is determined by \eqref{eq32}.

\begin{lemma}\label{lemma2}
With the linear approximation described in \eqref{eq32}, the objective function of problem {\bf P4\/} is an upper bound for that of problem {\bf P2\/}.
\end{lemma}

\begin{proof}
Due to the fact that the geometric mean is no greater than the arithmetic mean, we have $\sqrt[M]{{{{{\bf{\tilde n}}}_{{i_1}}}{{{\bf{\tilde n}}}_{{i_2}}} \cdots {{{\bf{\tilde n}}}_{{i_M}}}}} \le ({{{{{\bf{\tilde n}}}_{{i_1}}} + {{{\bf{\tilde n}}}_{{i_2}}} +  \cdots  + {{{\bf{\tilde n}}}_{{i_M}}}}})/{M}$. Since all elements of ${\bf{\tilde n}}$ are 0-1 variables, it can be easily verified that $\sqrt[M]{{{{{\bf{\tilde n}}}_{{i_1}}}{{{\bf{\tilde n}}}_{{i_2}}} \cdots {{{\bf{\tilde n}}}_{{i_M}}}}} = {{{\bf{\tilde n}}}_{{i_1}}}{{{\bf{\tilde n}}}_{{i_2}}} \cdots {{{\bf{\tilde n}}}_{{i_M}}}$. Thus, we have
\begin{align}\label{eq35}
{{{\bf{\tilde n}}}_{{i_1}}}{{{\bf{\tilde n}}}_{{i_2}}} \cdots {{{\bf{\tilde n}}}_{{i_M}}} \le \frac{{{{{\bf{\tilde n}}}_{{i_1}}} + {{{\bf{\tilde n}}}_{{i_2}}} +  \cdots  + {{{\bf{\tilde n}}}_{{i_M}}}}}{M}, \; \forall M \ge 2.
\end{align}
Applying \eqref{eq35} to all the higher-order expressions in \eqref{eq20}, ${\bf{c}}' {\bf{\tilde n}}$ is an upper bound to the objective function of {\bf P2\/}
\end{proof}

Based on Lemma~\ref{lemma2}, we further conclude that the optimal solution to problem {\bf P4\/} provides an upper bound to the optimal solution of {\bf P2\/}.

\subsubsection{Optimal Solution of {\bf P3\/} with Cutting Plane}

Since problem {\bf P3\/} has a linear objective function, the cutting plane method~\cite{Gomory58} can be used to derive the optimal solution. The idea of cutting plane is to find a plane that separates the non-integer solution from the polyhedron that satisfies the constraints and contains all integer feasible solutions.

Consider the polyhedron defined by ${\bf{A}} {\bf{\tilde n}} \le {\bf{b}}$,
determined by a combination of $\left( {K + 1} \right)\left( {J + 1} \right)$ inequalities as
\begin{align}\label{eq36}
{{\bf{a}}_i}{\bf{\tilde n}} \le {b_i},i = 1,2, ..., \left( {K + 1} \right)\left( {J + 1} \right),
\end{align}
where ${{\bf{a}}_i}$ is the $i$th row of ${\bf{A}}$. Let ${y_1}, {y_2}, ..., {y_{\left( {K + 1} \right)\left( {J + 1} \right)}} \ge 0$ and set
\begin{align}\label{eq37}
{{\bf{a}}^*} = \sum\limits_{i = 1}^{\left( {K + 1} \right)\left( {J + 1} \right)} {{y_i}{{\bf{a}}_i}}, \;\; {b^*} = \sum\limits_{i = 1}^{\left( {K + 1} \right)\left( {J + 1} \right)} {{y_i}{b_i}} .
\end{align}
Obviously, all solutions in the polyhedron ${\bf{A}} {\bf{\tilde n}} \le {\bf{b}}$ also satisfy ${{\bf{a}}^*}{\bf{\tilde n}} \le {b^*}$. If ${{\bf{a}}^*}$ is integral, i.e., all elements of ${{\bf{a}}^*}$ are integers, then all the {\em integer} solutions should satisfy
\begin{align}\label{eq38}
{{\bf{a}}^*}{\bf{\tilde n}} \le \left\lfloor {{b^*}} \right\rfloor,
\end{align}
where $\left\lfloor {{b^*}} \right\rfloor $ is the largest integer that is smaller than ${{b^*}}$. Then, \eqref{eq38} defines a cutting plane for {\bf P3\/}.

Adding an additional constraint described by the cutting plane, all the integer solutions are still included while some non-integer solutions are removed. Due to this property, we can first relax the integer constraint in solve {\bf P3\/} and solve the linear programming problem. If there are non-integer solutions, we add an additional constraint in the form of \eqref{eq38}. Then, we solve the linear programming problem with the updated constraints. If there are still non-integer elements in the solution vector, we continue to add another constraint following \eqref{eq38}. Such process terminates when all solution variables are integer. However, the effectiveness of this approach depends on the proper setting of parameters $y_i$. An efficient scheme to find the effective cutting plane was proposed in~\cite{Gomory58}; the details are omitted here due to lack of space.

\subsubsection{A Special Case without the Need for Cutting Plane}

We consider a special case with an additional condition. Suppose BS $j$ only receives a fixed amount of $L_j$ strongest multipath signals from each user and neglect the other multipath components with weaker signal strengths. We then have
\begin{align}\label{eq39}
{q_{1,j}} = {q_{2,j}} =  \cdots  = {q_{K,j}} = {L_j}, \; j = 0, 1, ...,J.
\end{align}
Note that, when $L_j$ is sufficiently large, this special case can be regarded
as the real case. Given \eqref{eq39}, we can divide both sides of \eqref{eq15} by $L_j$.  Then, ${\bf{A}}$ is updated by replacing ${\bf{q}}_j$ with
\begin{align}\label{eq40}
{{\bf{q}}_j} = {\left[ {1,1, ...,1} \right]^T}.
\end{align}
${\bf{b}}$ is updated as
\begin{align}\label{eq41}
&{\bf{b}}_{\left( {K + 1} \right) \left( {J + 1} \right) \times 1} \doteq \nonumber \\
&\left[ {\frac{{E_0}}{{L_0}}, ...,\frac{{E_J}}{{L_J}},1 \hspace{-0.025in}-\hspace{-0.025in} {x_{1,0}}, ..., 1 \hspace{-0.025in}-\hspace{-0.025in} {x_{K,0}},1 \hspace{-0.025in}-\hspace{-0.025in} {x_{1,1}}, ..., 1 \hspace{-0.025in}-\hspace{-0.025in} {x_{K,J}}} \right]^T.
\end{align}
It can be seen that there are exactly two $1$'s in each column of ${\bf{A}}$, with one from a column of ${\bf{Q}}$ and the other from a column of ${\bf{I}}$.
\begin{lemma}\label{lemma3}
${\bf{A}}$ is a totally unimodular matrix.
\end{lemma}
\begin{proof}
Omitted due to lack of space, a similar case can be found in~\cite{Feng16}.
\end{proof}
For a linear programming problem with a unimodular constraint matrix, all decision variables to the optimal solution are integers~\cite{Schrijver98}. Thus, the optimal solution of {\bf P3\/} can be obtained by relaxing the integer constraints and solving the resulting linear programming problem.

\begin{figure} [!t]
\centering
\includegraphics[width=3.4in]{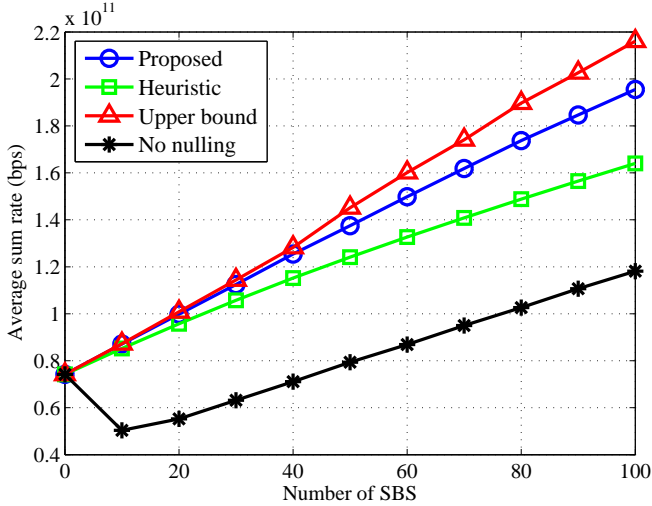}
\caption{Average sum rate versus number of SBS's. The number of users is 500.} 
\label{fig1}
\end{figure}

\begin{figure} [!t]
\centering
\includegraphics[width=3.4in]{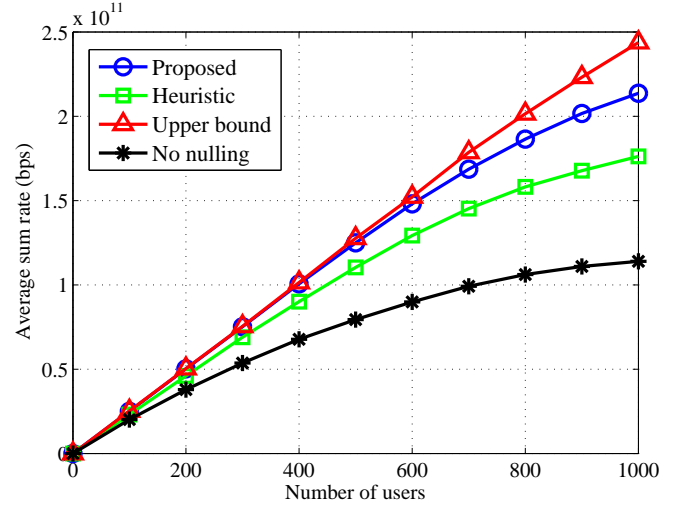}
\caption{Average sum rate versus average number of users. The number of SBS's is 50.} 
\label{fig2}
\end{figure}

\begin{figure} [!t]
\centering
\includegraphics[width=3.4in]{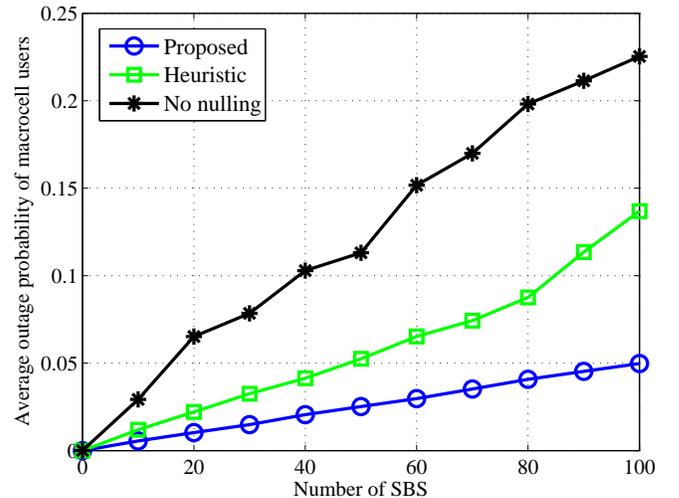}
\caption{Average outage probability of MUs versus number of SBS's. The number of users is 500.} 
\label{fig3}
\end{figure}

\section{Simulation Study \label{sec:sim}}

We validate the performance of the proposed scheme with Matlab simulations. We consider a macrocell overlaid with multiple small cells. The radii of the macrocell and a small cell are 1000 m and 50 m, respectively. Users in the coverage area of an SBS is served the SBS, while the others are served by the MBS. The macrocell and small cells share a total bandwidth of 4 MHz. The transmit power of the MBS is set to 40 dBm, while the transmit power of the MUEs has five levels ranging from 10 dBm to 30 dBm according to the distance between the MUE and MBS. The transmit power of an SBS and an SUE are set to 25 dBm and 15 dBm, respectively. We employ the ITU path loss model for both indoor and outdoor environments~\cite{itu}. The ratio $\frac{{{M_j} - {S_j} + 1}}{{{S_j}}}$ is set to be 100 for MBS and 10 for SBS's, respectively.

We consider a heuristic scheme for comparison (termed Heuristic). In the heuristic scheme, each BS chooses the user that causes the strongest interference and nullifies the signals of this user. This process is continued until the DoF constraint at a BS is violated. We also consider the case that no interference nulling is performed as a benchmark (termed No Nulling).

The sum rates of different schemes versus the number of SBS's are presented in Fig.~\ref{fig1}. It can be seen that without interference management, the sum rate first decreases as the SBS's are deployed, since the SINRs of MUEs are significantly reduced. Compared to the No Nulling scheme, a significant performance gain can be achieved by interference nulling, as a result of enhanced SINRs of both MUEs and SUEs. The sum rate with the proposed scheme is higher than the heuristic scheme since the proposed scheme optimizes the performance from the perspective of the entire network. We can also observe that the performance of the proposed scheme is close to the upper bound, indicating that the solution with linear approximation is near optimal. The performance gap between the proposed and heuristic scheme becomes larger as the number of SBS's is increased, since the heuristic scheme only achieves a local optimal solution for each BS. The resulting performance loss increases as the network gets larger.

In Fig.~\ref{fig2}, we compare the sum rates under different numbers of users. Due to the same reasons, similar trends are observed for the different schemes. It can be seen when the number of users is sufficiently large, the performance of proposed scheme is also affected by interference. This is because all the DoFs are used. Note that, the performance gap between the proposed scheme and the upper bound is increased when the number of users becomes large. This is because $P$ becomes smaller as $K$ is increased. Then the higher-order products are more likely to be $0$, and the linear approximation of the upper bound becomes more inaccurate.

Fig.~\ref{fig3} shows the outage performance of macrocell users (MU). We chose to evaluate the MUs since their average SINRs are lower, and hence they are more vulnerable to interference compared to small cell users. Due to the aggregated interference caused by SBS's to MUE, the average outage probability of MUs increases as the number of SBS's grows. It can be seen that with interference nulling performed by SBS's, the average outage probability of MUs is significantly reduced. When the number of SBS's gets large, the outage probabilities of proposed scheme and the heuristic scheme are also increased, since part of the DoFs are used to deal with interference between an SBS and SUEs served by other SBS's. The proposed scheme achieves the best performance, showing that the proposed approximation is accurate and achieves near optimal performance.

\section{Conclusions \label{sec:con}}

In this paper, we applied nested array in a massive MIMO HetNet and addressed the problem of interference nulling scheduling to maximize the sum rate of MUEs and SUEs. We formulated an integer programming problem and proposed an approximation solution algorithm, as well as a performance upper bound. The simulation results demonstrated the superior performance of the proposed scheme.

\section*{Acknowledgment}

This work is supported in part by the US National Science Foundation (NSF) under Grant CNS-1247955, and by the Wireless Engineering Research and Education Center (WEREC) at Auburn University.



\begin{thebibliography}{99}


\bibitem{Andrews14}
J. Andrews, S. Buzzi, W. Choi, S. V. Hanly, A. Lozano, A. C. K. Soong, and J. C. Zhang, ``What will 5G be?'' {\it IEEE J. Sel. Areas Commun.,\/} vol.32, no.6, pp. 1065--1082, June 2014.

\bibitem{Marzetta10}
T. L. Marzetta, ``Noncooperative cellular wireless with unlimited numbers of base station antennas,'' {\it IEEE Trans. Wireless Commun.,\/} vol.9, no.11, pp.3590--3600, Nov. 2010.

\bibitem{Ngo13}
H. Q. Ngo, E. G. Larsson, and T. L. Marzetta, ``Energy and spectral efficiency of very large multiuser MIMO systems,'' {\it IEEE Trans. Commun.,\/} vol.61, no.4, pp.1436--1449, Apr. 2013.

\bibitem{Feng16Network}
M. Feng and S. Mao, ``Harvest the potential of massive MIMO with multi-layer techniques,'' {\it IEEE Network\/}, to appear.

\bibitem{Xu14Access}
Y. Xu, G. Yue, and S. Mao, ``User grouping for massive MIMO in FDD systems: New design methods and analysis,'' {\it IEEE Access J.}, vol.2, no.1, pp.947--959, Sept. 2014.

\bibitem{Hosseini13}
K. Hosseini, J. Hoydis, S. ten Brink, and M. Debbah, ``Massive MIMO and small cells: How to densify heterogeneous networks,'' in {\it Proc. ICC'13,\/} Budapest, Hungary, June 2013, pp.5442--5447.

\bibitem{Bjornson13}
E. Bj$\ddot{o}$rnson, M. Kountouris, and M. Debbah, ``Massive MIMO and small cells: Improving energy efficiency by optimal soft-cell coordination,'' in {\it Proc. ICT'13,\/} Casablanca, Morocco, May 2013, pp.1--5.

\bibitem{Bethanabhotla14}
D. Bethanabhotla, O. Y. Bursalioglu, H. C. Papadopoulos, and G. Caire, ``User association and load balancing for cellular massive MIMO,'' in {\it Proc. IEEE Inf. Theory Appl. Workshop 2014,\/} San Diego, CA, Feb. 2014, pp. 1--10.

\bibitem{Xu15}
Y. Xu and S. Mao, ``User Association in Massive MIMO HetNets,'' {\em IEEE Systems J.}, to appear. DOI: 10.1109/JSYST.2015.2475702.

\bibitem{Feng16}
M. Feng, S. Mao, and T. Jiang, ``BOOST: Base station on-off switching strategy for energy efficient massive MIMO HetNets,'' in {\it Proc. IEEE INFOCOM'2016,\/} San Francisco, CA, Apr. 2016.

\bibitem{Feng14}
M. Feng, T. Jiang, D. Chen, and S. Mao, ``Cooperative small cell networks: High capacity for hotspots with interference mitigation,'' {\it IEEE Wireless Commun.,\/} vol. 21, no. 6, pp. 108--116, Dec. 2014.

\bibitem{Adhikary15}
A. Adhikary, H. S. Dhillon, and G. Caire, ``Massive-MIMO meets HetNet: Interference coordination through spatial blanking,'' {\it IEEE J. Select. Areas Commun.,\/} vol.33, no.6, pp. 1171--1186, June 2015.

\bibitem{Pal10}
P. Pal and P. P. Vaidyanathan, ``Nested arrays: A novel approach to array processing with enhanced degrees of freedom,'' {\it IEEE Trans. Sig. Proc.,\/} vol.58, no.8, pp.4167--4180, Aug. 2010.

\bibitem{Hoctor90}
R. T. Hoctor and S. A. Kassam, ``The unifying role of the coarray in aperture synthesis for coherent and incoherent imaging,'' {\it Proc. IEEE,\/} vol.78, no.4, pp.735--752, Apr. 1990.

\bibitem{Gomory58}
R. Gomory, ``Outline of an algorithm for integer solutions to linear programs,'' {\it Bull. Amer. Math. Soc.,\/} vol.64, no.5, pp. 275--278, Sept. 1958.

\bibitem{Schrijver98}
A. Schrijver, {\it Theory of Linear and Integer Programming,\/} Hoboken, NJ: John Wiley $\&$ Sons, 1998.
Hoboken, NJ

\bibitem{itu}
International Telecommunication Union, {\em Guidelines for Evaluation of Radio Transmission Technologies for IMT-2000}, Recommendation ITU-R M.1225, 1997.


%
%
%
%


\end{thebibliography}
\end{document}